\documentclass[a4paper,onecolumn,11pt]{quantumarticle}
\pdfoutput=1
\usepackage[utf8]{inputenc}
\usepackage[T1]{fontenc}

\usepackage[style=numeric, sorting=none]{biblatex} % Or other style
\usepackage{amsmath}
\usepackage{amsmath}
\usepackage{amssymb}
\usepackage{stfloats}
\usepackage{amsthm}
\usepackage{hyperref}
\usepackage{amsmath}
\usepackage{graphicx}
\usepackage{amssymb}
\usepackage{amsthm}
\usepackage{quantikz}
 \usepackage{float}
\newtheorem{lemma}{Lemma}

\usepackage{tikz}
\usepackage{lipsum}

\begin{document}

%\title{Hermitian matrix function synthesis via generalized quantum signal processing}

\title{Hermitian Matrix Function Synthesis without Block-Encoding}

%%%%
%%%%
\author{Anuradha Mahasinghe}
\email{anuradhamahasinghe@maths.cmb.ac.lk}
\affiliation{Centre for Quantum Information, Simulation and Algorithms, University of Western Australia, Australia}
\affiliation{Center for Mathematical Modeling, University of Colombo, Sri Lanka}
\orcid{0000-0002-2445-2701}
\author{Kaushika De Silva}
%\email{jingbo.wang@uwa.edu.au}
%\homepage{http://quantum-journal.org}
\orcid{0000-0003-0290-4698}
%\thanks{You can use the \texttt{\textbackslash{}email}, \texttt{\textbackslash{}homepage}, and \texttt{\textbackslash{}thanks} commands to add additional information for the preceding \texttt{\textbackslash{}author}. If applicable, this can also be used to indicate that a work has previously been published in conference proceedings.}
\affiliation{Centre for Quantum Information, Simulation and Algorithms, University of Western Australia, Australia}
\affiliation{University Paris City and University of Reunion, France}
%%%%
\author{Xavier Cadet}
\affiliation{Thayer School of Engineering, Darmouth College, USA}
%%%%
\author{Peter Chin}
\affiliation{Thayer School of Engineering, Darmouth College, USA}
%%%
\author{Frederic Cadet}
\affiliation{University Paris City and University of Reunion, France}
\affiliation{PEACCEL, AI for Biologics, Paris, France}
%%%%%
%%%%%
\author{Jingbo Wang}
\affiliation{Centre for Quantum Information, Simulation and Algorithms, University of Western Australia, Australia}
\email{jingbo.wang@uwa.edu.au}
%\homepage{http://quisa.tech}
\orcid{0000-0003-0290-4698}
% \thanks{You can use the \texttt{\textbackslash{}email}, \texttt{\textbackslash{}homepage}, and \texttt{\textbackslash{}thanks} commands to add additional information for the preceding \texttt{\textbackslash{}author}. If applicable, this can also be used to indicate that a work has previously been published in conference proceedings.}
% \affiliation{Centre for Quantum Information, Simulation and Algorithms, University of Western Australia, Australia}
%\author{Johannes Jakob Meyer}
%\affiliation{Dahlem Center for Complex Quantum Systems, Freie Universität Berlin, 14195 Berlin, Germany}
\orcid{0000-0003-1533-8015}
%\author{Victor V. Albert}
%\affiliation{Institute for Quantum Information and Matter \& Walter Burke Institute for Theoretical Physics, Caltech, Pasadena, CA 91125, USA}
%\orcid{0000-0002-0335-9508}
\maketitle

\begin{abstract}
Implementing polynomial functions of Hermitian matrices on quantum hardware is a foundational task in quantum computing, critical for accurate Hamiltonian simulation, quantum linear system solving, high-fidelity state preparation, machine learning kernels, and other advanced quantum algorithms. Existing state-of-the-art techniques, including Qubitization, Quantum Singular Value Transformation (QSVT), and Quantum Signal Processing (QSP), rely heavily on block-encoding the Hermitian matrix. These methods are often constrained by the complexity of preparing the block-encoded state, the overhead associated with the required ancillary qubits, or the challenging problem of angle synthesis for the polynomial's phase factors, which limits the achievable circuit depth and overall efficiency. In this work, we propose a novel and resource-efficient approach to implement arbitrary polynomials of a Hermitian matrix by leveraging the Generalized Quantum Signal Processing (GQSP) framework. Our method circumvents the need for block-encoding and avoids the compounding post-selection overheads characteristic of LCU-based constructions, achieving a stable, degree-independent success probability. We derive closed-form expressions for symmetric polynomial expansions and demonstrate how linear combinations of GQSP circuits can realize the desired transformation. This approach reduces resource overhead and opens new pathways for quantum algorithm design for functions of Hermitian matrices, particularly in settings where the Hermitian operator arises naturally from symmetric combinations 
of unitaries.
\end{abstract}

%In the \texttt{twocolumn} layout and without the \texttt{titlepage} option a paragraph without a previous section title may directly follow the abstract.
%In \texttt{onecolumn} format or with a dedicated \texttt{titlepage}, this should be avoided.

%Note that clicking the title performs a search for that title on \href{http://quantum-journal.org}{quantum-journal.org}.
%
%n this way readers can easily verify whether a work using the \texttt{quantumarticle} class was actually published in Quantum.
%If you would like to use \texttt{quantumarticle} for manuscripts not yet accepted in Quantum, or not even intended for submission to Quantum, please use the \texttt{unpublished} option to switch off all Quantum related branding and the hyperlink in the title.
%By default, this class also performs various checks to make sure the manuscript will compile well on the arXiv.
%If you do not intend to submit your manuscript to Quantum or the arXiv, you can switch off these checks with the \texttt{noarxiv} option.
%On the contrary, by giving the \texttt{accepted=YYYY-MM-DD} option, with \texttt{YYYY-MM-DD} the acceptance date, the note ``Accepted in Quantum YYYY-MM-DD, click title to verify'' can be added to the bottom of each page to clearly mark works that have been accepted in Quantum. 

\section{Introduction}
Quantum algorithms increasingly rely on the ability to implement functions of matrices on quantum hardware. Such transformations underpin key applications in quantum simulation, linear system solving, and spectral filtering. A central challenge is to synthesize matrix functions efficiently and accurately, using quantum circuits that respect hardware constraints and minimize resource overhead.
\par
Several quantum algorithmic frameworks have been developed to implement matrix functions on quantum computers. Among these, \textit{qubitization} stands out as a foundational and pioneering technique~\cite{LowChuang2016,low2019,Berry2015}. It enables the synthesis of Chebyshev polynomials of Hermitian matrices by embedding the target matrix $H$ (with $\|H\| \leq 1$) in a block-encoded unitary $U$, such that $\langle 0|U|0\rangle = H$. Through a structured sequence of controlled applications of $U$ and SU(2) rotations, a unitary transformation is obtained whose top-left block corresponds to $T_k(H)$, the $k$-th Chebyshev polynomial of the first kind.
\par
While qubitization is efficient for Chebyshev-type functions, many quantum applications require the implementation of general polynomial functions $P(H)$. This motivates the extension of qubitization via \textit{linear combination of unitaries} (LCU)~\cite{ChildsKothariSomma2017} leveraging the fact that any polynomial can be expressed in Chebyshev basis. In this approach, the desired polynomial is decomposed as $P(H) = \sum_k \alpha_k T_k(H)$, and each term $T_k(H)$ is synthesized using qubitization. The full transformation is then realized probabilistically via LCU, which requires ancilla qubits, controlled unitary selection, and post-selection. These requirements introduce significant resource overhead, particularly in terms of circuit depth and success probability, making LCU-based synthesis less practical for high-degree or complex-valued polynomials.
\par
Quantum Signal Processing (QSP) is the next intermediate step toward polynomial synthesis of matrices~\cite{low2017}. It is built upon the block-encoding paradigm in qubitization, extending its capabilities to synthesize arbitrary real or complex polynomials of Hermitian matrices. Given a Hermitian matrix $H$ with $\|H\| \leq 1$, the aim of QSP is constructing a unitary matrix whose top-left block gives $P(H)$, where $P$ is a target polynomial. This is achieved by the repeated application of the block-encoded unitary $U$ with a sequence of single-qubit SU(2) rotations, each parameterized by a phase angle $\phi_j$. The resulting circuit takes the form:
$
U_{\text{QSP}} = R(\phi_0) U R(\phi_1) U \cdots R(\phi_d),
$
where $R(\phi_j)$ denotes a rotation about the $Z$-axis, and $d$ is the degree of the polynomial.
\par
The synthesis of the angle sequence $\{\phi_j\}$ is highly nontrivial and typically involves solving a constrained optimization problem to match the desired polynomial profile on the unit circle~\cite{DongLinTong2021}. Moreover, QSP circuits bifurcate structurally depending on the parity of $\deg(P)$, requiring distinct constructions for even and odd degree polynomials~\cite{low2017}. This could lead to potential computational overhead, as two LCU calls are required to obtain the desired transformation for real polynomials; and for complex-valued polynomials, four LCU calls are required accordingly: two for real and two for imaginary. These constraints introduce a significant overhead in terms of ancilla qubits, circuit depth, and post-selection probability, making QSP less efficient for general-purpose polynomial synthesis despite its theoretical generality. 
Quantum singular value transformation (QSVT) is a generalization of the QSP framework to act on the singular values of block-encoded operators, thereby extending its applicability to non-square and non-Hermitian matrices~\cite{Gilyen2019QSVT,guo2024nonlineartransformationcomplexamplitudes}. While QSVT broadens the scope of matrix transformations, it inherits several structural limitations from QSP, including the complexity of angle synthesis, reliance on block-encoding, and the overhead introduced by linear combination of unitaries.
\par
One of the principal challenges in Quantum Signal Processing (QSP) lies in determining the phase angles required to implement the target polynomial transformation. This angle computation process is less efficient and often requires solving constrained optimization problems over the unit circle, especially for high-degree or complex-valued polynomials. \textit{Generalized Quantum Signal Processing (GQSP)} alleviates this burden by shifting part of the synthesis complexity onto the construction of a complementary polynomial, which is required to preserve the unitarity of the process~\cite{motlagh2024}. Given a normalized complex polynomial $P$, GQSP guarantees the existence of a complementary polynomial $Q$ that preserves uniatrity; i.e:
\[
|P(e^{i\theta})|^2 + |Q(e^{i\theta})|^2 = 1, \quad \forall \theta \in [0, 2\pi],
\]
and provides closed--form expressions for the SU(2) rotation parameters that implement the joint transformation. This structural shift reduces the dependency on numerical angle synthesis and replaces it with algebraic constraints on the polynomial pair $(P, Q)$, thus streamlining the circuit construction and improving scalability. 
The existence theorem for the complementary polynomial in GQSP is not merely existential--it is highly constructive as well. Embedded within it is an algorithmic procedure for generating the complementary polynomial $Q$ given a normalized complex polynomial $P$, typically involving the root-finding of associated complex polynomials on the unit circle~\cite{BerntsonSunderhauf2025}. While this approach is mathematically explicit, a more practical alternative is provided through a nonlinear optimization scheme that directly searches for the SU(2) rotation parameters satisfying the GQSP constraints, which offers promising efficiency and scalability~\cite{motlagh2024}, particularly when compared to the angle synthesis bottlenecks encountered in the QSP protocol.
\par
While GQSP provides a powerful framework for implementing polynomial transformations of unitary operators within quantum circuits, many practical applications demand the ability to act on more general classes of matrices---particularly Hermitian operators, which arise naturally in quantum simulation, optimization, and measurement. 
In many important cases, such Hermitian operators arise as the symmetric part of an underlying unitary, providing a natural bridge between GQSP and polynomial transformations of Hermitian matrices.
This raises a fundamental question: can one efficiently implement polynomial functions of such Hermitian matrices on a quantum computer?
\par
A tempting strategy is to apply the QSP paradigm directly to GQSP by block-encoding the Hermitian matrix into a unitary; however, this overlooks the structural distinctions between the two frameworks and fails to preserve the polynomial correspondence essential for extracting $P(H)$ in a usable form. Unlike QSP, where the polynomial $P(H)$ appears directly in the top-left block of the transformation due to the structure of the encoding, GQSP is intrinsically designed for unitary inputs and does not preserve this correspondence. As a result, the output state encodes $P(U)$, not $P(H)$, and extracting the desired transformation becomes nontrivial. Moreover, block-encoding itself incurs significant overhead in terms of ancilla qubits, circuit depth, and precision control.
\par
Addressing these limitations, we introduce a block‑encoding‑free framework for implementing polynomial functions of Hermitian matrices using a GQSP‑based quantum circuit. Our method bypasses the need for block‑encoding and avoids the practical overheads associated with post‑selection. This approach preserves the Hermitian structure, reduces resource costs, and extends the applicability of GQSP to a broader class of quantum algorithms. It is particularly effective for Hermitian operators that arise as the symmetric part of an underlying unitary—a situation that occurs frequently in quantum‑friendly, practically relevant settings.
\par
The remainder of the paper is organized as follows.  
Section~\ref{background} provides the necessary background on qubitization, quantum signal processing (QSP), and generalized quantum signal processing (GQSP). 
Section~\ref{algorithm} introduces our block-encoding-free algorithm based on symmetric polynomial expansion with the mathematical proof and the circuit construction. As any quantum algorithm for matrix function synthesis, ours is also non--universal, and usable in specific contexts. 
Section~\ref{applicability} discusses contexts of favorable applicability for our method.
Section~\ref{extensions} outlines possible future extensions of the framework to normal matrices. 
Finally, Section~\ref{conclusion} summarizes the contributions.

%\begin{figure}[t]
 % \centering
 % \includegraphics{example-plot.pdf}
 % \caption{Every figure must have an informative caption and a number.
  %  The caption can be placed above, below, or to the side of the figure, as you see fit.
   % The same applies for tables, boxes, and other floating elements. 
   % Quantum provides a Jupyter notebook to create plots that integrate seamlessly with \texttt{quantumarticle}, described in Section \ref{sec:plots}.
  %  Figures spanning multiple columns can by typeset with the usual \texttt{figure*} environment.}
 % \label{fig:figure1}
%\end{figure}
%See Fig.~\ref{fig:figure1} for an example of how to include figures.
%Feel free to place them at the top or bottom of the page, or in the middle of a paragraph as you see fit.
%Try to place them on the same page as the text referring to them.
%A figure on the first page can help readers remember and recognize your work more easily.

\section{Background}\label{background}
\subsection{Qubitization and Chebyshev Polynomial Encoding}
Let $A \in \mathbb{C}^{n \times n}$ be a Hermitian matrix with spectral norm $\|A\| \leq 1$ and consider $P(A)$, where $P$ is a polynomial or an analytic function. 
Qubitization constructs a unitary $U$ such that $\langle 0|U|0\rangle = A$, enabling the synthesis of Chebyshev polynomials of the first and second kind:
\[
T_k(x) = \cos(k \arccos x), \quad U_{k-1}(x) = \frac{\sin(k \arccos x)}{\sqrt{1 - x^2}}.
\]

These identities allow recursive construction of $T_k(A)$ and $U_{k-1}(A)$ via controlled applications of $U$ and its powers. However, qubitization is structurally limited to Chebyshev-type polynomials unless extended via Linear Combination of Unitaries (LCU), which introduces ancilla overhead and probabilistic post-selection. Furthermore, the block-encoding of $A$ must adhere to strict norm and sparsity requirements, which may be challenging to satisfy for practically useful Hermitian operators.
%~\cite{low2019}.

\subsection{Quantum Signal Processing (QSP)}
QSP generalizes qubitization by enabling the synthesis of arbitrary real or complex polynomials $P(A)$, provided $A$ is block-encoded in a unitary $U$. The QSP circuit applies a sequence of SU(2) rotations $R(\phi_j)$ interleaved with applications of $U$, yielding a unitary transformation $U_{\text{QSP}}$ whose top-left block approximates $P(A)$. For a degree-$d$ polynomial $P$, the circuit implements:
\[
U_{\text{QSP}} = R(\phi_0) U R(\phi_1) U \cdots R(\phi_d),
\]
where each $R(\phi_j) \in SU(2)$ is a single-qubit rotation. The synthesis of the angle sequence $\{\phi_j\}$ is nontrivial and typically requires solving a constrained optimization problem over the unit circle~\cite{low2017}. 
Recent work has also emphasized the need for numerically robust QSP constructions in black‑box settings, further illustrating the practical challenges inherent in QSP implementations~\cite{laneve2023robustblackboxquantumstatepreparation}. 
“In addition, efficient and stable evaluation of QSP phase factors has been studied in detail, further highlighting the computational challenges associated with angle synthesis~\cite{Dong_2021_efficient_phases_QSP}.
Furthermore, QSP circuits bifurcate structurally depending on the parity of $\deg(P)$, and complex-valued polynomials require doubling the LCU overhead, making implementation resource-intensive~\cite{childs2017}.

\subsection{Generalized Quantum Signal Processing (GQSP)}
Generalized Quantum Signal Processing (GQSP) extends the QSP framework by relaxing
the reliance on numerical angle synthesis. Instead of angles, GQSP introduces a \textit{complementary} polynomial to preserve unitarity, which has to be constructed from the target polynomial. More precisely, given a normalized target polynomial $P(z)$,
GQSP ensures the existence of a complementary polynomial $Q(z)$ such that
\[
|P(e^{i\theta})|^2 + |Q(e^{i\theta})|^2 = 1, \quad \forall \theta \in [0,2\pi].
\]
This condition allows the synthesis of unitary transformations whose matrix elements
encode $P$ and $Q$ simultaneously, thereby preserving unitarity while implementing
nontrivial polynomial transformations. It is noteworthy that the angle synthesis is straighforward, and all angles are computed via closed-form expressions.
Rotations with angles computed thus, together with the controlled unitaries finally result in the state 
$
\frac{1}{2}\ket{0}\,P(U)\ket{\psi}
\;+\;
\frac{1}{2}\ket{1}\,i\,Q(U)\ket{\psi}
$
for the input $\ket{0} \ket{\psi}  $.
%%%
\par
However, it is noteworthy that
constructing the complementary
polynomial $Q$ is not straightforward, but there are two approaches that makes the procedure less complicated.
The first approach goes parallel to the existence theorem, and accordingly $Q$ can be derived by solving for
the roots of associated polynomials on the unit circle. This method is mathematically
explicit and provides closed‑form rotation parameters, but can be computationally
intensive for high‑degree polynomials~\cite{BerntsonSunderhauf2025}.
The second approach is via a nonlinear optimization scheme that searches directly for the
SU(2) rotation parameters that satisfy the GQSP constraints. This approach is more
scalable and efficient in practice, particularly for large or complex polynomials~\cite{motlagh2024}.
Recent advances have also proposed robust numerical schemes for determining GQSP rotation parameters, further improving practical implementability~\cite{yamamoto2024robustanglefindinggeneralized}.
\par
By shifting the synthesis burden from angle computation to algebraic or optimization
constraints, GQSP streamlines circuit construction and improves scalability. This makes
it a powerful framework for polynomial transformations of unitary operators, and sets the
stage for our block‑encoding‑free extension to Hermitian matrices.

\section{The quantum algorithm}\label{algorithm}
\subsection{Symmetric polynomial expansion}
A key algebraic insight underlying our framework is that any Hermitian matrix $A$ with $\|A\| \leq 1$ admits a decomposition into a symmetric combination of unitary conjugates. 
such that
\begin{equation}
A = \frac{1}{2}(U + U^\dagger). 
\end{equation}
This identity reveals that Hermitian matrices can be embedded algebraically, without block-encoding, into symmetric expressions involving unitary matrices. Moreover, powers of $A$ can be expressed through symmetric expansions of powers of $U$ and $U^\dagger$, enabling polynomial constructions that preserve Hermiticity.
\par
This observation motivates our block-encoding-free approach: rather than embedding $A$ into a larger unitary, we construct polynomial functions $P(A)$ directly through symmetric combinations of conjugated unitaries. These unitaries are implemented via controlled applications of $e^{iA}$ and its conjugates, interleaved with SU(2) rotations derived from a complementary polynomial $Q$ satisfying the GQSP normalization condition. This structure allows us to synthesize $P(A)$ efficiently, without relying on block‑encoding or incurring its associated resource costs, and is especially effective when the Hermitian operator arises from symmetric combinations of unitaries.
\begin{lemma}[Symmetric Polynomial Expansion of Hermitian Powers]
Let $A \in \mathbb{C}^{n \times n}$ be a Hermitian matrix with $\|A\| \leq 1$. Then for every integer $n \geq 0$, there exists a degree-$n$ polynomial $R_n(x)$ such that
\[
A^n = R_n(U) + R_n(U^\dagger),
\]
for some uniatry $U$ and the polynomial $R_n(x)$ given by:
\[
R_n(x) = \left\{
\begin{aligned}
&\frac{1}{2^n} \sum_{k=0}^{\frac{n-1}{2}} \binom{n}{k} x^{n - 2k} && \text{if } n \text{ is odd}, \\
&\frac{1}{2^n} \sum_{k=0}^{\frac{n}{2}-1} \binom{n}{k} x^{n - 2k} + \frac{1}{2} \binom{n}{\frac{n}{2}} && \text{if } n \text{ is even}.
\end{aligned}
\right.
\]
\end{lemma}
%%%%%%
%%%%%%
%%%%%%%
%%%%%%%
\begin{proof}
Let $A$ be a Hermitian matrix with $\|A\| \leq 1$. 
Then, it is possible to exprerss $A$ as a sum of a unitary conjugate pair:
\begin{equation}
A= \frac{1}{2}(U + U^\dagger) . \label{eq:AfromU}
\end{equation}
We aim to show that for every integer $n \geq 0$, there exists a degree-$n$ polynomial $R_n(x)$ such that
\begin{equation}
A^n = R_n(U) + R_n(U^\dagger). \label{eq:main}
\end{equation}

\textbf{Base case 1:} $n = 0$. We have $A^0 = I$. From the definition,
\begin{equation}
R_0(x) = \frac{1}{2} \binom{0}{0} = \frac{1}{2}, \label{eq:R0}
\end{equation}
Therefore,
\begin{equation}
R_0(U) + R_0(U^\dagger) = \frac{1}{2} I + \frac{1}{2} I = I = A^0. \label{eq:base0}
\end{equation}

\textbf{Base case 2:} $n = 1$. We have $A^1 = A$. From the definition,
\begin{equation}
R_1(x) = \frac{1}{2} \binom{1}{0} x = \frac{1}{2} x, \label{eq:R1}
\end{equation}
Therefore,
\begin{equation}
R_1(U) + R_1(U^\dagger) = \frac{1}{2} U + \frac{1}{2} U^\dagger = \frac{1}{2}(U + U^\dagger) = A = A^1. \label{eq:base1}
\end{equation}
Both base cases hold. 
%%%%%%%
%%%%%%
%%%%%%%
%%%%%%%
\newline
\textbf{Inductive step:}
First suppose the result holds for odd $n$.
Then
\begin{equation*}
R_n(x) = \frac{1}{2^n} \sum_{k=0}^{n-1} h_k \binom{n}{k} x^{n - 2k}
\end{equation*}
and
$
A^n = R_n(U) + R_n(U^\dagger)
$.
Then,
\begin{align*}
A^{n+1} &= A \cdot [R_n(U) + R_n(U^\dagger)] \\
&= \frac{1}{2}(U + U^\dagger) \cdot [R_n(U) + R_n(U^\dagger)] \\
&= \frac{1}{2^{n+1}} (U + U^\dagger) \left[ \sum_{k=0}^{n-1}  \binom{n}{k} U^{n - 2k} + \sum_{k=0}^{n-1}  \binom{n}{k} (U^\dagger)^{n - 2k} \right] \\
&= \frac{1}{2^{n+1}} \sum_{k=0}^{n-1}  \binom{n}{k} \left[ U^{n - 2k + 1} + U^{n - 2k - 1} + 
(U^\dagger)^{n - 2k + 1} + (U^\dagger)^{n - 2k - 1} \right]
\end{align*}
Thus, $A^{n+1}$ also splits into two parts, a 
 \( U \)-part ($R_{n+1} (U)$) and a \( U^{\dagger} \)-part ($R_{n+1} (U^{\dagger})$) with same coefficients. Singling out the   \( U \)-part,
\begin{equation*}
R_{n+1} (U) = \sum_{k=0}^{n-1} \left[  \binom{n}{k} U^{n - 2k + 1} +  \binom{n}{k} U^{n - 2k - 1} \right]
\end{equation*}
which simplifies to:
%%%%
%%%%
\begin{align*}
\binom{n}{0} U^{n+1} + \left[\binom{n}{0} + \binom{n}{1}\right] U^{n-1} + \cdots + \left[ \binom{n}{\frac{n-3}{2}} + \binom{n}{\frac{n-1}{2}} \right] U^2 + \binom{n}{\frac{n-1}{2}}  I.
\end{align*}
\newline
By Pascal's identity, this reduces to:
\begin{align*}
\binom{n}{0} U^{n+1} + \binom{n}{1} U^{n-1} + \cdots + \binom{n}{\frac{n}{2}} U^2 + \binom{n}{\frac{n}{2}} I,
\end{align*}
which is equal to:
\begin{equation*}
\binom{n}{0} U^{n+1} + \sum_{k=1}^{n/2} \binom{n}{k} U^{n+1 - 2k} + \binom{n}{\frac{n-1}{2}} I.
\end{equation*}
\newline
Since \( \binom{n}{0} = \binom{n+1}{0} \) and \( \frac{n-1}{2} = \frac{n+1}{2} - 1 \), this is expressible as:
\begin{equation}
R_{n+1} (U) =\sum_{k=0}^{\frac{n+1}{2} - 1} \binom{n+1}{k} U^{n+1 - 2k} + \binom{n}{\frac{n-1}{2}} I.
\end{equation}
\newline
Applying Pascal's identity at the middle index \( k = \frac{n+1}{2} \):
\begin{equation*}
\binom{n+1}{\frac{n+1}{2}} = \binom{n}{\frac{n+1}{2}} + \binom{n}{\frac{n-1}{2}}.
\end{equation*}
\newline
Since \( n \) is odd and $n=\frac{n+1}{2} + \frac{n-1}{2} $, the two terms become symmetric binomial coefficients. That is, 
$
\binom{n}{\frac{n+1}{2}} = \binom{n}{\frac{n-1}{2}}.
$
Therefore, \( R_{n+1}(U) \) is expressible as,
\begin{equation}
R_{n+1}(U) = \sum_{k=0}^{\frac{n+1}{2}} \binom{n+1}{k} U^{n+1 - 2k} + \frac{1}{2} \binom{n+1}{\frac{n+1}{2}}.
\end{equation}
Thus, the result is established for \( n+1 \) when \( n \) is odd.
%%%%%
%%%%
%%%%
\par
Now consider the case $n$ is even. 
 That is,
\[
R_n(x) = \frac{1}{2^n} \sum_{k=0}^{\frac{n}{2}-1} \binom{n}{k} x^{n-2k}
\;+\; \frac{1}{2^n} \cdot \frac{1}{2} \binom{n}{\frac{n}{2}},
\]
and
$
A^n = R_n(U) + R_n(U^\dagger).
$
\begin{align*}
A^{n+1} &= A \cdot [R_n(U) + R_n(U^\dagger)]\\
&= \frac{1}{2}(U + U^\dagger) \cdot [R_n(U) + R_n(U^\dagger)].
\end{align*}
\newline
Substituting the expression for $R_n$, this is equal to,
\[
\frac{1}{2} (U + U^\dagger) \left[ 
\frac{1}{2^n} 
\sum_{k=0}^{\frac{n}{2}-1} \binom{n}{k} U^{n-2k} + 
\frac{1}{2^n} 
\sum_{k=0}^{\frac{n}{2}-1} \binom{n}{k} (U^{\dagger})^{n-2k} 
+ 
\binom{n}{n/2} I
\right]
\]
%%%%%
%%%%%
%%%%%
%\newline
Singling out the $U$-part of $A^{n+1}$:
\[
R_{n+1} (U) = \sum_{k=0}^{\frac{n}{2}-1} \left[ \binom{n}{k} U^{n+1-2k} + \binom{n}{k} U^{n-1-2k} \right]
+ \binom{n}{\frac{n}{2}} U.
\]
This is equal to:
\[
\binom{n}{0} U^{n+1}
+ \left[ \binom{n}{0} + \binom{n}{1} \right] U^{n-1}
+ \cdots
+ \left[\binom{n}{\frac{n}{2}-2} + \binom{n}{\frac{n}{2}-1} \right] U^3
+ \left[\binom{n}{\frac{n}{2}-1} + \binom{n}{\frac{n}{2}} \right] U.
\]
By Pascal's identity, this reduces to:
\[
\sum_{k=0}^{\frac{n}{2}} \binom{n+1}{k} U^{n+1 - 2k}.
\]
Similarly, the $U^\dagger$-part becomes:
\[
\sum_{k=0}^{\frac{n}{2}} \binom{n+1}{k} U^{\dagger\, (n+1 - 2k)}.
\]
Therefore,
$
A^{n+1} = R_{n+1}(U) + R_{n+1}(U^\dagger)
$
where,
\begin{align*}
R_{n+1}(x) = \frac{1}{2^{n+1}} \sum_{k=0}^{\frac{n}{2}} \binom{n+1}{k} x^{n+1 - 2k}.
\end{align*}
establishing the inductive step for even $n$ as well.
\end{proof}
%%%%
%%%%
%%%
Now it is an immediate consequence that for any target matrix polynomial with coefficients $c_n \in \mathbb{C}$, we have
$
P(A) = \sum_{n=0}^d c_n A^n. \label{eq:PAexpand}
$
By the lemma, each power satisfies
\begin{equation}
A^n = R_n(U) + R_n(U^\dagger), \label{eq:lemmaapply}
\end{equation}
where $R_n$ is the degree-$n$ polynomial defined previously. Therefore,
\begin{equation}
P(A) = \sum_{n=0}^d c_n \big(R_n(U) + R_n(U^\dagger)\big). \label{eq:PAfinal}
\end{equation}
This shows that any polynomial function of a Hermitian matrix $A$ can be expressed as a symmetric combination of $R_n(U)$ and $R_n(U^\dagger)$. The structure of the polynomial $R_n(x)$
 follows directly from binomial expansions, with coefficients drawn from Pascal’s Triangle according to parity, allowing immediate symbolic construction. 
%%%%%
%%%%%
\newpage
\subsection{Circuit analysis}
\begin{figure}[H]
  \begin{center}
  % (your quantikz code here)
\begin{quantikz}
\lstick{$\ket{0}$}      & \gate{H} & \ctrl{1} & \octrl{1} & \gate{H} & \meter{} & \qw \\
\lstick{$\ket{0}$}      & \qw      & \gate[wires=2]{\parbox{3.6cm}{\centering $\sum_{j=0}^{d} c_j R_j(U)$\\ \text{via GQSP}}} & \gate[wires=2]{\parbox{3.6cm}{\centering $\sum_{j=0}^{d} c_j R_j(U^{\dagger})$\\ \text{via GQSP}}} & \qw & \meter{} & \qw \\
\lstick{$\ket{\psi}$}   & \qw      & \qw      & \qw      & \qw      & \qw      & \rstick{$P(A)\ket{\psi}$}
\end{quantikz}
\end{center}
  %%%
  \caption{Quantum circuit implementing the polynomial transformation 
  $\sum_{j=0}^d c_j R_j(U)$ via GQSP, controlled by the ancilla qubit. 
  Measurement of the first two qubits projects the data qubit into 
  $P(A)\ket{\psi}$.}
  \label{fig:GQSP-circuit}
\end{figure}
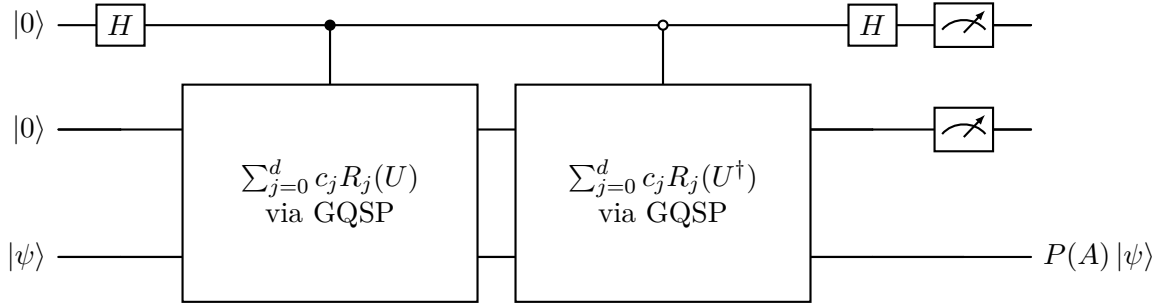
%%%
%%%
The circuit in Figure~\ref{fig:GQSP-circuit} implements the polynomial transformation 
$P(A)\ket{\psi} $ which is given by $ \left( \sum_{j=0}^d c_j R_j(U) + \sum_{j=0}^d c_j R_j(U^{\dagger}) \right)\ket{\psi}$ via controlled applications 
of the unitary $U = A + i\sqrt{I - A^2}$, using the GQSP framework. The first two qubits 
serve as ancilla and control registers, initialized in the state $\ket{0}$ and subjected to 
Hadamard gates to create superposition. The controlled blocks apply the polynomial 
components $R_j(U)$ weighted by the coefficients $c_j$, and the measurement of the ancilla 
qubits at zero projects the data qubit into the state ${P}(A)\ket{\psi}$.
\par
Suppose $\widetilde{P}$ denotes the polynomial $\sum_{j=0}^{d} c_j R_j(x)$ that will be applied on $U$ and $U^{\dagger}$ via GQSP. Then, $P(A)=\widetilde{P}(U)+\widetilde{P}(U^{\dagger})$. Also suppose $\widetilde{Q}$ is the complementary polynomial in this GQSP implementation. Our circuit initiates from the state
$
 \ket{0} \ket{0} \ket{\psi}
$
 after the first Hadamard, the composite state becomes
$
\tfrac{1}{\sqrt{2}} \ket{0} \ket{0} \ket{\psi}
+ \tfrac{1}{\sqrt{2}} \ket{1} \ket{0} \ket{\psi}.
$
This goes through two controlled (one open, one closed, controlled by first qubit) GQSP operations, applied to other qubits, resulting in the state:
\[
\frac{1}{2} \ket{0} \ket{0} \, \widetilde{P}(U) \ket{\psi}
+ \frac{1}{2} \ket{0} \ket{1} \, i \widetilde{Q}(U) \ket{\psi}
+ \frac{1}{2} \ket{1} \ket{0} \, \widetilde{P}(U^\dagger) \ket{\psi}
+ \frac{1}{2} \ket{1} \ket{1} \, i \widetilde{Q}(U^\dagger) \ket{\psi}.
\]
The final Hadamard on the first qubit thus makes the state
\[
\frac{1}{2\sqrt{2}} \left[
\ket{0} \ket{0} \left( \widetilde{P}(U) + \widetilde{P}(U^\dagger) \right) \ket{\psi}
+ \sum_{(a,b) \in \{0,1\}^2 \setminus \{(0,0)\}} \ket{a} \ket{b} \, F_{ab}(U, U^\dagger) \ket{\psi}
\right],
\]
where,
$
F_{01}(U,U^\dagger)=i\big(\widetilde{Q}(U)+\widetilde{Q}(U^\dagger)\big)
$,
$
F_{10}(U,U^\dagger)=\widetilde{P}(U)-\widetilde{P}(U^\dagger),
$
and
$
F_{11}(U,U^\dagger)=i\big(\widetilde{Q}(U)-\widetilde{Q}(U^\dagger)\big).
$
\par
The joint state contains symmetric and antisymmetric combinations of $\widetilde{P}(U)$ and $\widetilde{P}(U^\dagger)$ together with the $\widetilde{Q}$ terms. By measuring the two ancilla qubits and post‑selecting the outcome $\ket{00}$, the data register collapses to
$
\big(\widetilde{P}(U)+\widetilde{P}(U^\dagger)\big)\ket{\psi}.
$
The probability of obtaining the desired post‑selection outcome $\ket{00}$ is given by the squared norm of the projected state,
$ \frac{1}{8}\,\big\|\big(\widetilde{P}(U)+\widetilde{P}(U^\dagger)\big)\ket{\psi}\big\|^2.
$
%This probability depends on both the input state $\ket{\psi}$ and the synthesis coefficients $c_n$, and quantifies the efficiency of the circuit implementation.
%%%
%%%
\section{Domains of practical applicability}\label{applicability}
The applicability of our framework is governed by a structural principle: a Hermitian operator is particularly amenable to our method when it can be represented as a real symmetric polynomial in a unitary. This situation arises naturally in many settings, 
including Toeplitz and circulant operators, finite–difference Laplacians, long–range lattice interactions, and Szegedy–type quantum walks, where the underlying unitary generator is explicitly available. Beyond these native cases, our approach also extends 
to Hermitian operators that admit a constructive unitary dilation, allowing them to be embedded algebraically into the same symmetric–unitary form. Together, these regimes cover a broad class of practically relevant operators for which polynomial synthesis can 
be performed without block encoding or matrix square–root evaluation, relying only on GQSP applied to an appropriate unitary primitive.
%%%%

%%%%%
%%%%
\subsection{Hermitian operators with unitary generators}
A particularly favorable domain for our method occurs when the Hermitian operator \(H\) is generated from an underlying unitary through the relation \(H = \tfrac12(U + U^\dagger)\) or its higher‑order generalizations. 
More generally, any Hermitian operator expressible as a real symmetric polynomial in a unitary,
\begin{equation}\label{genform}
A = a_0 I + \sum_{k=1}^{n} a_k \left( U^k + U^{-k} \right),
\end{equation}
falls directly within the applicability domain of our method, since GQSP operates natively on the unitary \(U\) and its powers.
This structure appears naturally in several quantum walk and lattice‑based models, where the unitary \(U\) is the primitive operator governing the dynamics. In such cases, polynomial synthesis can be performed directly on \(U\), eliminating the need for block encoding or square‑root evaluation.
%%%%%
\subsubsection{Toeplitz and circulant operators}
Toeplitz and circulant matrices form a broad and practically important class of operators that naturally admit representations as symmetric polynomials in a unitary. Let \(S\) denote the periodic shift operator on an \(N\)-dimensional 
Hilbert space, \(S\ket{j} = \ket{j+1 \bmod N}\). Any circulant matrix \(C\) can be written as a polynomial in \(S\),
\[
C = \sum_{k=0}^{N-1} c_k S^{k},
\]
and is therefore diagonalized by the discrete Fourier transform. When the coefficients satisfy \(c_k = c_{N-k}\), the matrix is Hermitian and admits the 
symmetric representation
\[
C = c_0 I + \sum_{k=1}^{m} c_k\left(S^{k} + S^{-k}\right),
\]
which is precisely a real symmetric polynomial in the unitary shift operator \(S\).

Such operators arise in a wide range of applications, including convolution kernels, signal processing transforms, discretized differential operators, and translation-invariant Hamiltonians on periodic lattices. Their spectral properties are determined by the trigonometric polynomial
\[
\lambda(\theta) = c_0 + 2\sum_{k=1}^{m} c_k \cos(k\theta),
\qquad \theta = \frac{2\pi j}{N},
\]
which is the Fourier symbol of the Toeplitz or circulant operator.

Since \(C\) is a symmetric polynomial in the unitary \(S\), polynomial 
transformations \(P(C)\) can be implemented directly by applying GQSP to the underlying shift operator. No block encoding is required, and no matrix square root appears in the construction. More generally, any Toeplitz operator with finite bandwidth and symmetric coefficients,
\[
T = a_0 I + \sum_{k=1}^{m} a_k\left(S^{k} + S^{-k}\right),
\]
lies squarely within the applicability domain of our framework. This includes circulant Hamiltonians, convolution operators, and a wide class of translation-invariant models encountered in quantum simulation and numerical analysis.
%%%%%
\subsubsection{Laplacians from finite difference discretizations}
A highly practical example arises from the finite–difference discretization of partial differential equations. Consider the one dimensional heat equation
\[
\frac{\partial u(x,t)}{\partial t}
=
\kappa \frac{\partial^2 u(x,t)}{\partial x^2},
\]
defined on a spatial grid consisting of $N$ uniformly spaced points
$x_j = j\Delta x$, $j=0,\dots,N-1$, subject to periodic Neumann boundary conditions:
\[
u(0,t) = u(L,t), \qquad 
\partial_x u(0,t) = \partial_x u(L,t).
\]
Once this is discretized using standard central finite–difference approximation,
\[
\frac{\partial^2 u}{\partial x^2}(x_j)
\approx
\frac{u_{j+1}-2u_j+u_{j-1}}{\Delta x^2},
\]
the discretized Laplacian takes the circulant form \(H = \mathrm{circ}(-1,2,-1,0,\ldots,0)\), which can be written compactly as \(H = \tfrac{1}{\Delta x^{2}}(S + S^\dagger - 2I)\) in terms of the periodic shift operator \(S\). Thus \(H\) is a symmetric polynomial in the unitary \(S\), and the natural signal unitary for our framework is simply \(U = S\). Writing \(\widetilde{H} = \tfrac{1}{\Delta x^{2}}(S + S^\dagger)\), we have \(H = \widetilde{H} + a_{0}I\) with \(a_{0} = -2/\Delta x^{2}\). The operators \(H\) and \(\widetilde{H}\) share the same eigenvectors, and their eigenvalues differ only by the constant shift \(a_{0}\). Consequently, for any scalar function \(f\) we may define \(g(x) = f(x + a_{0})\) and obtain \(f(H) = g(\widetilde{H})\). The synthesis of \(f(H)\) therefore reduces to synthesizing \(g(\widetilde{H})\), where \(\widetilde{H}\) is of the form \(\tfrac12(U + U^\dagger)\) up to scaling. This places the discretized Laplacian squarely within the ideal applicability regime of our method.
%%%%%%%%%%%%%%
%%%%%%%%%%%%%%%%%
\subsubsection{Szegedy-type quantum walks}
Another significant instance of the structure 
\(H = \tfrac12(U + U^\dagger)\) arises in Szegedy-type quantum walks. 
Given a reversible Markov chain with transition matrix \(P\), let 
\(\Pi\) and \(\Pi'\) denote the orthogonal projectors onto the subspaces 
spanned by the states 
\(\{\ket{i}\ket{\psi_i}\}\) and \(\{\ket{\psi_i}\ket{i}\}\), respectively, 
where \(\ket{\psi_i} = \sum_j \sqrt{P_{ij}}\,\ket{j}\). 
The Szegedy walk operator is the unitary
\[
U_{\mathrm{Sz}} = (2\Pi - I)(2\Pi' - I),
\]
a product of reflections that encodes the transition structure of the 
Markov chain.

A central object associated with this walk is the discriminant operator
\[
D = \tfrac12\left(U_{\mathrm{Sz}} + U_{\mathrm{Sz}}^\dagger\right),
\]
which is Hermitian and whose spectrum is directly related to that of the 
Markov chain. In particular, if \(\lambda\) is an eigenvalue of \(P\), then 
\(\cos(\theta)\) with \(\theta = \arccos(\lambda)\) appears as an eigenvalue 
of \(D\). Thus, spectral transformations of \(P\)---such as spectral 
filtering, gap amplification, or polynomial acceleration---can be carried 
out by applying polynomial transformations to \(D\).

Since \(D\) is precisely the symmetric part of the unitary 
\(U_{\mathrm{Sz}}\), it lies squarely within the ideal applicability regime 
of our framework. Polynomial transformations \(P(D)\) can be synthesized 
directly by applying GQSP to the underlying unitary \(U_{\mathrm{Sz}}\), 
without constructing a block encoding of \(D\) or evaluating any matrix 
square roots. This provides a resource-efficient pathway for implementing 
a wide class of quantum walk–based algorithms, including those involving 
Markov chain discriminants, reflection-based operators, and reversible 
stochastic processes.
%%%
%%%
\subsubsection{Long–range lattice interactions}
Long–range couplings on a one–dimensional periodic lattice provide another natural 
instance of Hermitian operators expressible as symmetric polynomials in a unitary. 
Let \(S\) denote the periodic shift operator, \(S\ket{j} = \ket{j+1 \bmod N}\). 
While nearest–neighbour interactions correspond to the operator 
\(\tfrac12(S + S^\dagger)\), longer–range hopping processes are captured by the 
higher–order shifts \(S^{k}\) and \(S^{-k}\). A prototypical example is the Hamiltonian
\[
H = \frac12\left(S^{k} + S^{-k}\right),
\]
which couples lattice sites separated by \(k\) positions. In matrix form, this 
corresponds to the Toeplitz  operator
\[
H_{i,j}=
\begin{cases}
\frac12 & |i-j|=k \pmod N,\\
0 & \text{otherwise}.
\end{cases}
\]
Such Hamiltonians arise naturally in long–range hopping models, extended tight–binding 
systems, and spin chains with nonlocal interactions. Since \(S\) is unitary, the operator 
\(H\) is a real symmetric polynomial in \(S\), and therefore lies directly within the 
applicability domain of our framework. More generally, any long–range Hamiltonian of the 
form

\[
H = a_0 I + \sum_{k=1}^{m} a_k\left(S^{k} + S^{-k}\right)
\]

is a symmetric polynomial in the unitary \(S\). Polynomial transformations \(P(H)\) can 
thus be implemented by applying GQSP to the underlying shift operator \(S\), without 
constructing a block encoding of \(H\) or evaluating any matrix square roots. This makes 
long–range lattice models a particularly suitable class of operators for our 
block–encoding–free approach.

%Further operators of the same structural type include discriminant operators, reflection based unitaries, Markov chain embeddings, multi‑unitary symmetric expansions, Toeplitz and circulant operators, graph Laplacians, and long‑range interaction models.
%%%%
\subsection{Hermitian operators with algebraic unitary extensions}
While the algorithm helps in synthesizing polynomials of a class of matrices that naturally decompose into a sum of unitaries, its applicability extends further. Note that the unitary operator $U$ in the symmetric expansion are expressible as $
U = A + i\sqrt{I - A^2}$
whose Hermitian conjugate is
$U^\dagger = A - i\sqrt{I - A^2}$.
Computing matrix square roots can be highly challenging and, in the worst case, even infeasible in practice—this difficulty persists even when the original matrix is sparse. However, the situation here leads to a more favorable setting: since \(A\) is Hermitian, the operator \(I - A^2\) is also Hermitian. Consequently, \(I - A^2\) admits a complete set of orthonormal eigenvectors with real eigenvalues. The square root operation then reduces to applying 
$
\sqrt{\,1 - \lambda_i^2\,}
$
to each eigenvalue \(\lambda_i\). This procedure is straightforward whenever the spectrum of \(A\) is known or bounded. Therefore, unlike the general case where matrix square roots may be problematic, the Hermitian structure ensures that the square root implementation here remains tractable and well-behaved.
%%%%
%%%%
\par
On the other hand, there exist many practical settings in which the implementation of the square root is efficient. 
In particular, this occurs when \(A\) is a sparse Hermitian matrix with a bounded spectrum, or when \(A\) arises from structured operators such as graph Laplacians or Hamiltonians with known eigenstructure~\cite{mahasinghe2016toeplitz,zhou2017circulant,mahasinghe2019vibration}. 
\par
For structured Hermitian matrices such as tridiagonal operators, there exist 
numerically efficient algorithms to apply functions of the form 
\(f(I - A^{2})\), including \(f(x) = \sqrt{x}\), to vectors, without forming 
the full dense matrix \(\sqrt{I - A^{2}}\). This makes the corresponding 
unitary operator
$
U = A + i\sqrt{I - A^{2}}
$
practically accessible in such structured settings.

%Moreover, when \(A\) is sparse with non--negative entries concentrated near the main diagonal, the sparsity pattern of \(\sqrt{I - A^2}\) is often preserved. 
%Consequently, our circuit offers a distinct advantage for polynomial synthesis in the case of tridiagonal Hermitian matrices, where both structure and sparsity can be exploited to achieve efficient implementation.
\par
It is noteworthy that graph Laplacians form a natural class of Hermitian matrices where the square root implementation 
is particularly efficient. By construction, Laplacians are sparse with non--negative entries concentrated near the diagonal, and their spectral properties are well understood: eigenvalues 
are non--negative and often clustered according to graph structure. As a result, the operator 
$I - A^2$ inherits sparsity and positivity, and its square root reduces to applying  $\sqrt{1 - \lambda_i^2}$ to each eigenvalue $\lambda_i$. This spectral tractability makes 
graph Laplacians especially suitable for our circuit, since polynomial synthesis can exploit both sparsity and spectral bounds to achieve efficient implementation.
\par
Another favorable setting arises when the Hermitian matrix $A$ admits a low--rank approximation. 
In such cases, the action of $\sqrt{I - A^2}$ is nontrivial only on a small subspace spanned by the dominant eigenvectors of $A$, while the remainder of the spectrum contributes negligibly. 
This property allows the square root to be implemented with reduced resource overhead, since only a limited number of eigenvalues require explicit transformation to $\sqrt{1 - \lambda_i^2}$. 
Our circuit therefore gains a specific advantage in low--rank contexts, where polynomial synthesis can be confined to a compact subspace. This efficiency is particularly relevant in applications 
such as data analysis and machine learning, where matrices often exhibit low effective rank due to redundancy or compression.
%%%%
%%%%
\par
Beyond Fourier-diagonalizable operators, constructive unitary dilations are also practically accessible for structured Hermitian matrices such as tridiagonal or more generally banded operators. In these cases, although the matrix 
\(\sqrt{I - A^{2}}\) is typically dense, functions of the form \(f(I - A^{2})\) can be applied efficiently to vectors using Krylov subspace techniques, Chebyshev  approximations, or other matrix-function algorithms that exploit the banded 
structure of \(A\). This enables the unitary operator 
$
U = A + i\sqrt{I - A^{2}}
$
to be implemented without forming the square root explicitly, thereby extending  the applicability of our framework to a broad class of banded Hermitian operators.
\par
Further, the square root term in our construction coincides with the  Halmos dilation, a canonical unitary embedding of a contraction. It has been widely employed as a subroutine for block-encoding~\cite{arrazola2023lcu, camps2022explicit, sunderhauf2024blockencoding}, and several researchers 
have demonstrated its efficiency in contexts such as sparse Hermitian matrices and structured 
operators. Since our circuit leverages the same square root structure, those favorable contexts 
extend naturally to our setting as well.
%%%%
%%%%
%\subsection{Contexts where block-encoding is prohibitive}
\subsection{Contexts where block-encoding is prohibitive}
 Block-encoding has become a standard technique for embedding matrices into unitary operators, 
and it serves as a powerful subroutine in many quantum algorithms. However, its implementation 
often requires a significant number of ancillary qubits and complex resource management, 
particularly when the target matrix lacks convenient structure. Even in favorable sparse cases, explicit block‑encoding circuits can be highly structured and resource‑intensive~\cite{camps2022explicit}. Several researchers have noted that block-encoding can become prohibitively expensive, with gate 
complexity and ancilla requirements dominating the overall algorithmic cost ~\cite{camps2022explicit,wu2024, sunderhauf2024blockencoding,rullkotter2025variational}. 
The prohibitive cost of block-encoding arises from its reliance on ancillary registers, 
nontrivial state preparation, and gate complexity that often scales unfavorably with matrix 
structure%\cite{kane2025bosons,galvao2025variational}. 
~\cite{kane2025bosons}. 
In many cases, these overheads dominate the algorithm, making block-encoding 
impractical for near-term devices.
%%%%
%%%%
\par
Our circuit circumvents these difficulties by avoiding the need for 
ancillary registers and the associated normalization circuitry. Moreover, although our method involves a post-selection step, its success probability remains constant and degree-independent, unlike block-encoding-based LCU constructions where repeated post-selection causes the success amplitude to shrink multiplicatively. This stability in post-selection behavior provides a practical advantage in settings where block-encoding is either infeasible or prohibitively expensive.
%%%%%
%%%%%
\subsection{LCU overheads in polynomial synthesis}
Linear Combination of Unitaries (LCU) is a powerful framework for implementing polynomial 
transformations, but its cost can be prohibitive in practice. In particular, within the 
Quantum Singular Value Transformation (QSVT) framework, real and complex parts of a polynomial 
must often be treated separately, since standard constructions restrict achievable polynomials 
to real coefficients and definite parity~\cite{sunderhauf2023generalized,motlagh2024}. 
This separation increases resource overhead and complicates circuit synthesis. By contrast, 
our circuit directly implements the Halmos dilation $U = A + i\sqrt{I - A^2}$, which naturally 
accommodates complex polynomials without splitting into real and imaginary components. Crucially, the post-selection in our circuit occurs only 
once, and its success probability does not degrade with the degree of 
the polynomial or the number of terms in its expansion. This stands in sharp 
contrast to LCU-based approaches, where repeated post-selection leads to 
exponentially decreasing success amplitudes. As a result, our framework offers 
a resource-efficient alternative for polynomial synthesis in regimes where LCU 
overheads dominate.
%%%
%%%%
\section{Outlook and extensions}\label{extensions}
While our framework is tailored to Hermitian matrices, the underlying algebraic structure suggests 
a natural extension to normal matrices. Since normal operators admit unitary diagonalization and 
preserve spectral properties, it may be possible to construct analogous symmetric polynomial 
expansions using conjugated unitaries derived from the polar decomposition or other canonical 
forms. This would broaden the applicability of our circuit to a wider class of quantum transformations, 
including those arising in non-Hermitian dynamics, open quantum systems, and certain classes of 
measurement operators. Exploring this generalization could reveal new structural identities and 
circuit architectures beyond the Hermitian setting.
\par
Many quantum algorithms rely on rational approximations to implement functions such as matrix 
inversion, resolvents, and fractional powers. Recent work has used block-encoding to simulate 
these functions via LCU or Chebyshev expansions, which  incurs a significant overhead. Our circuit, 
by avoiding block-encoding and operating directly on the Hermitian input, may be extendible to 
rational function synthesis through recursive polynomial constructions or controlled feedback 
mechanisms. This opens the possibility of implementing rational approximations with reduced 
resource cost.
\par
An alternative approach is to consider the exponential unitary operator 
\(U = e^{iA}\), which satisfies the identity
$
U + U^\dagger = 2\cos(A).
$
In this formulation, a target polynomial transformation \(P(A)\) can be rewritten as
\(
P(A) = P\!\left(\arccos\!\left(\tfrac{1}{2}(U + U^\dagger)\right)\right),
\)
so that implementing \(P(A)\) reduces to synthesizing a polynomial approximation of the 
\(\arccos\) function applied to the Hermitian operator \(\tfrac{1}{2}(U + U^\dagger)\)~\cite{YusenDiscussion}.
A systematic treatment of such higher-order operator transformations, and specifically their implementation in quantum computation, can be found in  the framework of higher-order quantum operations~\cite{taranto2025higherorder}. A detailed exploration of this particular direction is reserved for future work.

%%%
\section{Conclusion}\label{conclusion}
We introduced a block-encoding-free framework for synthesizing polynomial functions of Hermitian matrices using the Generalized Quantum Signal Processing (GQSP) paradigm. Our approach departs from the traditional reliance on block-encoding and and multi-stage LCU constructions, offering instead a direct and resource efficient method for implementing $P(A)$ when the Hermitian operator $A$ admits a symmetric-unitary representation or a constructive unitary dilation. This structural shift enables polynomial 
transformations to be carried out using only controlled applications of a suitable unitary primitive and closed-form GQSP rotations, thereby avoiding the substantial overheads associated with amplitude preparation, ancilla management and post-selection.
\par
The core of our method is the symmetric polynomial expansion that expresses every power $A^{n}$ as a sum of two unitary polynomial components $R_{n}(U)$ and $R_{n}(U^\dagger)$. This expansion enables the construction of arbitrary polynomials $P(A)$ through linear combinations of GQSP circuits applied to $U$ and 
$U^\dagger$, without embedding $A$ in a larger unitary. 
The resulting circuit requires only a single post-selection step whose success probability remains stable and independent of the polynomial degree, in sharp contrast with LCU-based approaches where repeated post-selection leads to rapidly diminishing amplitudes.
\par
We also identified two broad applicability regimes for our method. The first consists of Hermitian operators that arise natively as real symmetric polynomials in a unitary, including Toeplitz and circulant matrices, finite-difference Laplacians, Szegedy-type 
quantum walks, and long-range lattice Hamiltonians. In these settings, the underlying unitary generator is explicit, and polynomial synthesis can be performed directly on the primitive unitary. The second regime includes Hermitian operators that admit a constructive unitary dilation, allowing them to be embedded algebraically into the same 
symmetric-unitary form. 
%This extends the reach of our framework to structured operators such as tridiagonal or banded matrices, where matrix-function techniques enable efficient 
%application of the unitary without forming it explicitly.
\par
Together, these results demonstrate that GQSP provides a powerful and versatile tool for Hermitian matrix function synthesis, capable of addressing practical limitations of block-encoding-based methods while maintaining circuit simplicity and scalability. Our 
framework opens new avenues for quantum algorithm design in Hamiltonian simulation, spectral filtering, quantum walks, and linear system solving, particularly in settings where the target operator possesses exploitable algebraic or structural properties.
\par]
Future work includes extending the symmetric-unitary expansion to normal matrices, developing optimized compilation strategies for GQSP circuits in hardware-constrained architectures, and exploring hybrid schemes that combine our approach with existing block-encoding techniques when advantageous. We anticipate that the ideas introduced 
here will contribute to a broader rethinking of matrix-function synthesis on quantum hardware, emphasizing algebraic structure and resource efficiency as guiding principles.

\section*{Acknowledgment}
This project was supported by the Australian Government through its Critical Technologies Challenge Program (CTCP). The authors would like to thank Yusen Wu, Tal Gurfinkel, Jack Blyth, James Greenwell, and Jie Pan for valuable discussions.

\onecolumn\newpage
%\appendix

%\section{First section of the appendix}
%Quantum allows the usage of appendices.

%\subsection{Subsection}
%Ideally, the command %\texttt{\textbackslash{}appendix} should be put before the appendices to get appropriate section numbering.
%The appendices are then numbered alphabetically, with numeric (sub)subsection numbering.
%Equations continue to be numbered sequentially.
%\begin{equation}
 % A \neq B
%\end{equation}
%You are free to change this in case it is more appropriate for your article, but a consistent and unambiguous numbering of sections and equations must be ensured.

%If you want your appendices to appear in \texttt{onecolumn} mode but the rest of the document in \texttt{twocolumn} mode, you can insert the command \texttt{\textbackslash{}onecolumn\textbackslash{}newpage} before \texttt{\textbackslash{}appendix}.   

%\section{Problems and Bugs}
%In case you encounter problems using the quantumarticle class please analyze the error message carefully and look for help online; \href{http://tex.stackexchange.com/}{http://tex.stackexchange.com/} is an excellent resource.
%If you cannot resolve a problem, please open a bug report in our bug-tracker under \href{https://github.com/quantum-journal/quantum-journal/issues}{https://github.com/quantum-journal/quantum-journal/issues}.
%You can also contact us via email under \href{mailto:latex@quantum-journal.org}{latex@quantum-journal.org}, but it may take significantly longer to get a response.
%In any case, we need the full source of a document that produces the problem and the log file showing the error to help you.

\printbibliography

\end{document}